\documentclass[a4paper,numberwithinsect, 11pt]{article}

\usepackage[margin=1in]{geometry}
\usepackage{hyperref}
\usepackage{breakurl}

\usepackage[USenglish]{babel}

\usepackage{microtype}

\usepackage{enumitem}

\usepackage{algorithm}
\usepackage[noend]{algpseudocode}

\usepackage{xspace}

\usepackage{amsmath,amssymb,amsthm}

\theoremstyle{plain}
\newtheorem{theorem}{Theorem}[section]
\newtheorem{lemma}[theorem]{Lemma}

\newtheorem{prop}[theorem]{Proposition}

\newcommand{\nats}{\mathbb{N}}
\newcommand{\ints}{\mathbb{Z}}
\newcommand{\field}{\mathbb{F}}

\newcommand{\Oh}{{\cal O}}
\newcommand{\tOh}{\tilde{\cal O}}

\newcommand{\poly}{\mathrm{poly}}

\newcommand{\fMM}{f_\mathrm{MM}}

\newcommand{\NSETH}{\textsc{NSETH}\xspace}

\newcommand{\BMM}{\textsc{BMM}\xspace}
\newcommand{\osMM}{\textsc{os-MM}\xspace}
\newcommand{\OV}{\textsc{OV}\xspace}
\newcommand{\APSP}{\textsc{APSP}\xspace}
\newcommand{\AllZeroes}{\textsc{AllZeroes}\xspace}
\newcommand{\MPV}{\textsc{MM-Verification}\xspace}
\newcommand{\MPC}{\textsc{MM-Correction}\xspace}
\newcommand{\THREESUM}{\textsc{3SUM}\xspace}
\newcommand{\UPIT}{\textsc{UPIT}\xspace}

\newcommand{\FindNonzero}{\textsc{FindNonzero}\xspace}
\newcommand{\UpdateVals}{\textsc{UpdateValuesAndList}\xspace}

\newcommand{\todoremark}[2]{%
}

\title{On Nondeterministic Derandomization of Freivalds' Algorithm: Consequences, Avenues and Algorithmic Progress} 

\author{Marvin K\"unnemann\thanks{Max Planck Institute for Informatics, Saarland Informatics Campus, Saarbr\"ucken, Germany. \texttt{marvin@mpi-inf.mpg.de}}}

\date{}

\begin{document}

\maketitle

\begin{abstract}
Motivated by studying the power of randomness, certifying algorithms and barriers for fine-grained reductions, we investigate the question whether the multiplication of two $n\times n$ matrices can be performed in near-optimal \emph{nondeterministic} time $\tOh(n^2)$. Since a classic algorithm due to Freivalds verifies correctness of matrix products probabilistically in time $\Oh(n^2)$, our question is a relaxation of the open problem of derandomizing Freivalds' algorithm.

We discuss consequences of a positive or negative resolution of this problem and provide potential avenues towards resolving it. Particularly, we show that sufficiently fast deterministic verifiers for \THREESUM or univariate polynomial identity testing yield faster deterministic verifiers for matrix multiplication. Furthermore, we present the partial algorithmic progress that distinguishing whether an integer matrix product is correct or contains between 1 and $n$ erroneous entries can be performed in time $\tOh(n^2)$ -- interestingly, the difficult case of deterministic matrix product verification is not a problem of ``finding a needle in the haystack'', but rather cancellation effects in the presence of many errors. 

Our main technical contribution is a deterministic algorithm that corrects an integer matrix product containing at most $t$ errors in time $\tOh(\sqrt{t} n^2 + t^2)$. To obtain this result, we show how to compute an integer matrix product with at most $t$ nonzeroes in the same running time. This improves upon known deterministic output-sensitive integer matrix multiplication algorithms for $t = \Omega(n^{2/3})$ nonzeroes, which is of independent interest.
\end{abstract}

\section{Introduction}

Fast matrix multiplication algorithms belong to the most exciting algorithmic developments in the realm of low-degree polynomial-time problems. Starting with Strassen's polynomial speedup~\cite{Strassen69} over the naive $\Oh(n^3)$-time algorithm, extensive work (see, e.g.,~\cite{CoppersmithW90,VassilevskaW12,LeGall14}) has brought down the running time to $\Oh(n^{2.373})$ (we refer to~\cite{Blaeser13} for a survey). This leads to substantial improvements over naive solutions for a wide range of applications; for many problems, the best known algorithms make crucial use of fast multiplication of square or rectangular matrices. To name just a few examples, we do not only obtain polynomial improvements for numerous tasks in linear algebra (computing matrix inverses, determinants, etc.),  graph theory (finding large cliques in graphs~\cite{NesetrilP85}, All-Pairs Shortest Path for bounded edge-weights~\cite{AlonGM97}), stringology (context free grammar parsing~\cite{Valiant75}, RNA folding and language edit distance~\cite{BringmannGSW16}) and many more, but also strong subpolynomial improvements such as a $2^{\Omega(\sqrt{\log n})}$-factor speed-up for the All-Pairs Shortest Path problem (\APSP)~\cite{Williams14} or similar improvements for the orthogonal vectors problem (\OV)~\cite{AbboudWY15}. 
It is a famous open question whether the matrix multiplication exponent $\omega$ is equal to  $2$.

Matrix multiplication is the search version of the \MPV problem: given $n\times n$ matrices $A,B$ and a candidate $C$ for the product matrix, verify whether $AB=C$. There is a surprisingly simple randomized algorithm due to Freivalds~\cite{Freivalds79} that is correct with probability at least $1/2$: Pick a random vector $v\in \{0,1\}^n$, compute the matrix-vector products $Cv$ and $A(Bv)$, and declare $AB = C$ if and only if $Cv = ABv$. Especially given the simplicity of this algorithm and the widely-shared hope that $\omega=2$, one might conjecture that a deterministic version of Freivalds' algorithm exists. Alas, while refined ways to pick the random vector $v$ reduce the required number of random bits to $\log n + \Oh(1)$ \cite{NaorN93, KimbrelS93}, a $\tOh(n^2)$-time deterministic algorithms for matrix product verification remains elusive.

The motivation of this paper is the following question: 
\begin{center}
{\itshape  Can we solve Boolean, integer or real matrix multiplication in \emph{nondeterministic} $\tOh(n^2)$ time? }
\end{center}

Here we say that a functional problem $f$ is in nondeterministic time $t(n)$ if $f$ admits a \emph{$t(n)$-time verifier}: there is a function $v$, computable in deterministic time $t(n)$, where $n$ denotes the problem size of $x$, such that for all $x,y$ there exists a certificate $c$ with $v(x, y, c) = 1$ if and only $y = f(x)$.\footnote{Throughout the paper, we view any decision problem $P$ as a binary-valued functional problem. Thus a $t(n)$-time verifier for $P$ shows that $P$ is in nondeterministic \emph{and} co-nondeterministic time $t(n)$.}

Note that a $\tOh(n^2)$-time derandomization of Freivalds' algorithm would yield an affirmative answer: guess $C$, and verify $AB=C$ using the deterministic verification algorithm. In contrast, a nondeterministic algorithm may guess additional information, a \emph{certificate} beyond a guess $C$ on the matrix product, and use it to verify that $C=AB$. Surprising faster algorithms in such settings have recently been found for 3SUM and all problems subcubic equivalent to APSP under deterministic reductions \cite{CarmosinoGIMPS16}; see \cite{VassilevskaWilliamsW10,VassilevskaW18} for an overview over subcubic equivalences to APSP. 

In this paper, we discuss consequences of positive or negative resolutions of this question, propose potential avenues for an affirmative answer and present partial algorithmic progress. In particular, we show that (1) sufficiently fast verifiers for \THREESUM or univariate polynomial identity testing yield faster nondeterministic matrix multiplication algorithms, (2) in the integer case we can detect existence of between 1 and $n$ erroneous entries in $C$ in deterministic time $\tOh(n^2)$ and (3) we provide a novel deterministic output-sensitive integer matrix multiplication algorithm that improves upon previous deterministic algorithms if $AB$  has at least $n^{2/3}$ nonzeroes. 

\subsection{Further Motivation and Consequences}

Our motivation stems from studying the power of randomness, as well as algorithmic applications in certifiable computation, and consequences for the fine-grained complexity of polynomial-time problems.

\emph{Power of Randomness}: Matrix-product verification has one of the simplest randomized solution for which no efficient derandomization is known -- the currently best known deterministic algorithm simply computes the matrix product $AB$ in deterministic time $\Oh(n^\omega)$ and checks whether $C=AB$. Exploiting nondeterminism instead of randomization may yield insights into when and under which conditions we can derandomize algorithms without polynomial increases in the running time.

A very related case is that of univariate polynomial identity testing (\UPIT): it has a similar status with regards to randomized and deterministic algorithms. As we will see, finding $\tOh(n^2)$-time nondeterministic derandomizations for \UPIT is a more difficult problem, so that resolving our main question appears to be a natural intermediate step towards nondeterministic derandomizations of \UPIT, see Section~\ref{sec:intro-relationships}. 

\emph{Practical Applications -- Deterministic Certifying Algorithms}:
Informally, a \emph{certifying algorithm} for a functional problem $f$ is an algorithm that computes, for each input $x$, besides the desired output $y=f(x)$ also a certificate $c$ such that there is a \emph{simple} verifier that checks whether $c$ proves that $y=f(x)$ indeed holds~\cite{McConnellMNS11}. If we fix our notion of \emph{simplicity} to be that of being computable by a fast deterministic algorithm, then our notion of verifiers turns out to be a suitable notion to study existence of certifying algorithms -- it only disregards the running time needed to compute the certificate $c$. 

Having a fast verifier for matrix multiplication would certainly be desirable -- while Freivalds' algorithm yields a solution that is sufficient for many practical applications, it can never \emph{completely} remove doubts on the correctness. Since matrix multiplication is a central ingredient for many problems, 
fast verifiers for matrix multiplication imply fast verifiers for many more problems.

In fact, even if $\omega=2$, finding \emph{combinatorial}\footnote{Throughout this paper, we call an algorithm \emph{combinatorial}, if it does not use sophisticated algebraic techniques underlying the fastest known matrix multiplication algorithms.} strongly subcubic verifiers  is of interest, as these are more likely to yield practical advantages over more naive solutions. In particular, the known subcubic verifiers for all problems subcubic equivalent to APSP (under deterministic reductions)~\cite{CarmosinoGIMPS16}  all rely on fast matrix multiplication, and might not yet be relevant for practical applications.

\emph{Barriers for SETH-based Lower Bounds}:
Given the widely-shared hope that $\omega=2$, can we rule out conditional lower bounds of the form $n^{c-o(1)}$ with $c>2$ for matrix multiplication, e.g., based on the Strong Exponential Time Hypothesis (SETH)~\cite{ImpagliazzoP01}? 
Carmosino et al.~\cite{CarmosinoGIMPS16} proposed the Nondeterministic Strong Exponential Time Hypothesis (\NSETH) that effectively postulates that there is no $\Oh(2^{(1-\varepsilon)n})$-time co-nondeterministic  algorithm for $k$-SAT for all constant $k$. Under this assumption, we can rule out  fast nondeterministic or co-nondeterministic algorithms for all problems that have \emph{deterministic} fine-grained reductions from $k$-SAT. Conversely, if we find a nondeterministic matrix multiplication algorithm running in time $n^{c+o(1)}$, then \NSETH implies that there is no SETH-based lower bound of $n^{c'-o(1)}$, with $c'>c$, for matrix multiplication using deterministic reductions.

\emph{Barriers for Reductions in Case of a \emph{Negative} Resolution}:
Suppose that there is a negative resolution of our main question, specifically that Boolean matrix multiplication has no $n^{c-o(1)}$-time verifier for some $c>2$ (observe that this would imply $\omega>2$). Then by a simple $\Oh(n^2)$-time nondeterministic reduction from Boolean matrix multiplication to triangle finding (implicit in the proof of Theorem~\ref{thm:intro-MPVto3SUM} below) and a known $\Oh(n^2)$-time reduction from triangle finding to Radius~\cite{AbboudGW15}, Radius has no $n^{c-o(1)}$-time verifier. This state of affairs would rule out certain kinds of subcubic reductions from Radius to Diameter, e.g., deterministic many-one-reductions, since these would transfer a simple $\Oh(n^2)$-time verifier for Diameter\footnote{We verify that a graph $G$ has diameter $d$ as follows: For every vertex $v$, we guess the shortest path tree originating in $v$. It is straightforward to use this tree to verify that all vertices $v'$ have distance at most $d$ from $v$ in time $\Oh(n)$. Thus, we can prove that the diameter is at most $d$ in time $\Oh(n^2)$. For the lower bound, guess some vertex pair $u,v$ and verify that their distance is indeed $d$ using a single-source shortest path computation in time $\Oh(m+n\log n)= \Oh(n^2)$.} to Radius. Note that finding a subcubic reduction from Radius to Diameter is an open problem in the fine-grained complexity community~\cite{AbboudGW15}. 

\subsection{Structural Results: Avenues Via Other Problems}
\label{sec:intro-relationships}

We present two particular avenues for potential subcubic or even near-quadratic matrix multiplication verifiers: finding fast verifiers for either \THREESUM or univariate polynomial identity testing.

\paragraph*{3SUM}
One of the core hypotheses in the field of hardness in P is the \THREESUM problem~\cite{GajentaanO95}. Despite the current best time bound of $\Oh(n^2\cdot \frac{\poly \log \log n}{\log^2 n})$~\cite{BaranDP08, Chan18} being only slightly subquadratic, recently a strongly subquadratic verifier running in time $\tOh(n^{3/2})$ was found~\cite{CarmosinoGIMPS16}. We have little indication to believe that this verification time is optimal; for the loosely related computational model of decision trees, a remarkable near-linear time bound has been obtained just this year~\cite{KaneLM18}. 

By a simple reduction, we obtain that any polynomial speedup over the known \THREESUM verifier yields a subcubic Boolean matrix multiplication verifier.  
In particular, establishing a near-linear \THREESUM verifier would yield a positive answer to our main question in the Boolean setting.

\begin{theorem}\label{thm:intro-MPVto3SUM}
Any $\Oh(n^{3/2-\varepsilon})$-time verifier for \THREESUM yields a $\Oh(n^{3-2\varepsilon})$-time verifier for Boolean matrix multiplication.
\end{theorem}

Under the \BMM hypothesis, which asserts that there is no combinatorial $\Oh(n^{3-\varepsilon})$-time algorithm for Boolean matrix multiplication (see, e.g.,~\cite{AbboudVW14}), a $n^{3/2-o(1)}$-time lower bound (under randomized reductions) for combinatorial \THREESUM algorithms is already known \cite{JafargholiV16, VassilevskaWilliamsW10}. The above result, however, establishes a stronger, non-randomized relationship between the verifiers' running times by a simple proof exploiting nondeterminism.   

\paragraph*{UPIT}
Univariate polynomial identity testing (\UPIT) asks to determine, given two degree-$n$ polynomials $p,q$ over a finite field of polynomial order, represented as arithmetic circuits with $\Oh(n)$ wires, whether $p$ is identical to $q$.
By evaluating and comparing $p$ and $q$  at $n+1$ distinct points or $\tOh(1)$ random points, we can solve \UPIT deterministically in time $\tOh(n^2)$ or with high probability in time $\tOh(n)$, respectively. A nondeterministic derandomization, more precisely, a $\Oh(n^{2-\varepsilon})$-time verifier,  would have interesting consequences~\cite{Williams16}: it would refute the Nondeterministic Strong Exponential Time Hypothesis posed by Carmosino et al.~\cite{CarmosinoGIMPS16}, which in turn would prove novel circuit lower bounds, deemed difficult to prove. We observe that a sufficiently strong nondeterministic derandomization of \UPIT would also give a faster matrix multiplication verifier.

\begin{theorem}
Any $\Oh(n^{3/2-\varepsilon})$-time verifier for \UPIT yields a $\Oh(n^{3-2\varepsilon})$-time verifier for integer matrix multiplication. 
\end{theorem}

Note that this avenue might seem more difficult to pursue than a direct attempt at resolving our main question, due to its connection to NSETH and circuit lower bounds. Alternatively, however, we can view the specific arithmetic circuit obtained in our reductions as an interesting intermediate testbed for ideas towards derandomizing \UPIT. In fact, our algorithmic results were obtained by exploiting the connection to \UPIT, and exploiting the structure of the resulting specialized circuits/polynomials.

\subsection{Algorithmic Results: Progress on Integer Matrix Product Verification}

Our main result is partial algorithmic progress towards the conjecture in the integer setting. Specifically, we consider a restriction of \MPV to the case of detecting a bounded number~$t$ of errors. Formally, let $\MPV_t$ denote the following problem: given $n\times n$ integer matrices $A,B, C$ with polynomially bounded entries, produce an output ``$C=AB$'' or ``$C\ne AB$'', where the output must always be correct if $C$ and $AB$ differ in at most $t$ entries.   

Our main result is an algorithm that solves $\MPV_t$ in near-quadratic time for $t = \Oh(n)$ and in strongly subcubic time for $t = \Oh(n^{c})$ with $c<2$.
\begin{theorem}\label{thm:main-detect}
For any $1 \le t \le n^2$, $\MPV_t$ can be solved deterministically in time $O( (n^2 +tn)\log^{2+o(1)} n)$. 
\end{theorem}

Interestingly, this shows that detecting the presence of very few errors is not a difficult case. Instead of a needle-in-the-haystack problem, we rather need to find a way to deal with cancellation effects in the presence of at least $\Omega(n)$ errors. 

As a corollary, we obtain a different near-quadratic-time randomized algorithm for \MPV than Freivalds' algorithm: Run the algorithm of Theorem~\ref{thm:main-detect} for $t=n$ in time $\tOh(n^2)$. Afterwards, either $C=AB$ holds or $C$ has at least $\Omega(n)$ erroneous entries. Thus it suffices to sample $\Theta(n)$ random entries $i,j$ and to check whether $C_{i,j} = (AB)_{i,j}$ for all sampled entries (by naive computation of $(AB)_{i,j}$ in time $\Oh(n)$ each) to obtain an $\tOh(n^2)$-time algorithm that correctly determines $C=AB$ or $C\ne AB$ with constant probability. 
Potentially, this alternative to Freivalds' algorithm might be simpler to derandomize.

Finally, our algorithm for \emph{detecting} up to $t$ errors can be extended to a more involved algorithm that also \emph{finds} all erroneous entries (if no more than $t$ errors are present) and \emph{correct} them in time $\tOh(\sqrt{t} n^2 + t^2)$. In fact, this problem turns out to be equivalent to the notion of output-sensitive matrix multiplication $\osMM_t$: Given $n\times n$ matrices $A,B$ of polynomially bounded integer entries with the promise that $AB$ contains at most $t$ nonzeroes, compute $AB$.  

\begin{theorem}\label{thm:main-correct}
Let $1\le t \le n^2$.
 Given $n\times n$ matrices $A,B,C$ of polynomially bounded integers, with the property that $C$ differs from $AB$ in at most $t$ entries, we can compute $AB$ in time $\Oh(\sqrt{t}n^2 \log^{2+o(1)} n + t^2 \log^{3+o(1)} n)$. 
 Equivalently, we can solve $\osMM_t$ in time $\Oh(\sqrt{t}n^2 \log^{2+o(1)} n + t^2 \log^{3+o(1)} n)$.
 \end{theorem}

Previous work by Gasieniec et al.~\cite{GasieniecLLPT17} gives a $\tOh(n^2 + tn)$ \emph{randomized} solution, as well as a $\tOh(tn^2)$ deterministic solution. Because of the parameter-preserving equivalence between~$t$ error correction and $\osMM_t$, this task is also solved by the randomized $\tOh(n^2 + tn)$-time algorithm due to Pagh~\cite{Pagh13}\footnote{For $t= \omega(n)$, Jacob and St\"ockel~\cite{JacobS15} give an improved randomized $\tOh(n^2 (t/n)^{\omega-2})$-time algorithm.} and the deterministic $\Oh(n^2+t^2 n \log^5 n)$-time algorithm due to Kutzkov~\cite{Kutzkov13}. Note that our algorithm improves upon Kutzkov's algorithm for $t = \Omega(n^{2/3})$, in particular, our algorithm is strongly subcubic for $t = \Oh(n^{3/2-\varepsilon})$ and even improves upon the best known fast matrix multiplication algorithm for $t = \Oh(n^{0.745})$.

\subsection{Further Related Work}

There is previous work that claims to have resolved our main question in the affirmative. Unfortunately, the approach is flawed; we detail the issue in the appendix. Furthermore, using the unrealistic assumption that integers of bit length $\tOh(n)$ can be multiplied in constant time, Korec and Wiedermann~\cite{KorecW14} provide an $\Oh(n^2)$-time deterministic verifier for integer matrix multiplication.     

Other work considers $\MPV$ and $\osMM$ in quantum settings, e.g.,~\cite{BuhrmanS06,JefferyKM12}. Furthermore, better running times can be obtained if we restrict the distribution of the errors over the guessed matrix/nonzeroes over the matrix product: Using rectangular matrix multiplication, Iwen and Spencer~\cite{IwenS09} show how to compute $AB$ in time $\Oh(n^{2+\varepsilon})$ for any $\varepsilon >0$, if no column (or no row) of $AB$ contains more than $n^{0.29462}$ nonzeroes. Furthermore, Roche~\cite{Roche18} gives a randomized algorithm refining the bound of Gasieniec et al.~\cite{GasieniecLLPT17} using, as additional parameters, the total number of nonzeroes in $A,B,C$ and the number of distinct columns/rows containing an error.

For the setting of Boolean matrix multiplication, several output-sensitive algorithms are known~\cite{SchnorrS98,YusterZ05,AmossenP09,Lingas09}, including a simple deterministic $\Oh(n^2 + tn)$-time algorithm~\cite{SchnorrS98} and, exploiting fast matrix multiplication, a randomized $\tOh(n^2 t^{\omega/2-1})$-time solution~\cite{Lingas09}. Note that in the Boolean case, our parameter-preserving reduction from error correction to output-sensitive multiplication (Proposition~\ref{prop:MPVtoAllZeroes}) no longer applies, so that these algorithms unfortunately do not immediately yield error correction algorithms.

\subsection{Paper Organization}
After collecting notational conventions and introducing polynomial multipoint evaluation as our main algorithmic tool in Section~\ref{sec:prelim}, we give a high-level description over the main ideas behind our results in Section~\ref{sec:mainIdeas}. We prove our structural results in Section~\ref{sec:relations}. Our first algorithmic result on error detection is proven in Section~\ref{sec:errordetection}. The main technical contribution, i.e., the proof of Theorem~\ref{thm:main-correct}, is given in Section~\ref{sec:errorcorrection}. We conclude with open questions in Section~\ref{sec:open}.

\section{Preliminaries}
\label{sec:prelim}

Recall the definition of a $t(n)$-time verifier for a functional problem $f$: there is a function $v$, computable in deterministic time $t(n)$ with $n$ being the problem size of $x$, such that for all $x,y$ there exists a certificate $c$ with $v(x, y, c) = 1$ if and only $y = f(x)$. Here, we assume the word RAM model of computation with a word size $w=\Theta(\log n)$.

For $n$-dimensional vectors $a, b$ over the integers, we write their inner product as $\langle a, b \rangle = \sum_{k=1}^n a[k]\cdot b[k]$, where $a[k]$ denotes the $k$-th coordinate of $a$. For any matrix $X$, we write $X_{i,j}$ for its value at row $i$, column $j$. We typically represent the $n\times n$ matrix $A$ by its $n$-dimensional row vectors $a_1, \dots, a_n$, and the $n\times n$ matrix $B$ by its $n$-dimensional column vectors $b_1, \dots, b_n$ such that $(AB)_{i,j} = \langle a_i, b_j \rangle$. For any $I\subseteq [n], J \subseteq [n]$, we obtain a submatrix $(AB)_{I,J}$ of $AB$ by deleting from $AB$ all rows not in $I$ and all columns not in $J$.

\paragraph*{Fast Polynomial Multipoint Evaluation}

Consider any finite field $\field$ and let $M(d)$ be the number of additions and multiplications in $\field$ needed to multiply two degree-$d$ univariate polynomials. Note that $M(d) = \Oh(d \log d \log \log d) = \Oh(d \log^{1+o(1)} n)$, see, e.g.~\cite{vzGathenG13}. 

\begin{lemma}[{Multipoint Polynomial Evaluation~\cite{Fiduccia72}}]\label{lem:polymultipoint}
Let $\field$ be an arbitrary field. Given a degree-$d$ polynomial $p \in \field[X]$ given by a list of its coefficients $(a_0, \dots, a_d)\in \field^{d+1}$, as well as input points $x_1, \dots, x_d \in \field$, we can determine the list of evaluations $(p(x_1), \dots, p(x_d))\in \field^{n}$ using $\Oh( M(d) \log d)$ additions and multiplications in $\field$.
\end{lemma}
Thus, we can evaluate $p$ on any list of inputs $x_1,\dots,x_n$ in time $\Oh( (n+d) \log^{2+o(1)} d)$.

\section{Technical Overview}
\label{sec:mainIdeas}

We first observe a simple parameter-preserving equivalence of the following problems,
\begin{description}
\item[$\MPV_t$] Given $\ell\times n, n \times \ell, \ell\times \ell$ matrices $A,B,C$ such that $AB$ and $C$ differ in $0 \le z \le t$ entries, determine whether $AB=C$, i.e., $z=0$,
\item[$\AllZeroes_t$] Given $\ell\times n, n \times \ell$ matrices $A,B$ such that $AB$ has $0 \le z \le t$ nonzeroes, determine whether $AB=0$, i.e., $z=0$.
\end{description}
We also obtain a parameter-preserving equivalence of their ``constructive'' versions,
\begin{description}
\item[$\MPC_t$] Given $\ell\times n, n \times \ell, \ell\times \ell$ matrices $A,B,C$ such that $AB$ and $C$ differ in $0 \le z \le t$ entries, determine~$AB$,
\item[$\osMM_t$] Given $\ell\times n, n \times \ell$ matrices $A,B$ such that $AB$ has $0 \le z \le t$ nonzeroes, determine~$AB$.
\end{description}
For any problem $P_t$ among the above, let $T_P(n,\ell,t)$ denote the optimal running time to solve~$P_t$ with parameters $n$, $\ell$ and $t$.

\begin{prop}\label{prop:MPVtoAllZeroes}
Let $\ell \le n$ and $1\le t \le n^2$. We have 
\begin{align*}
 T_\MPV(n, \ell, t) & = \Theta(T_\AllZeroes(n,\ell,t))\\
 T_\MPC(n,\ell, t) & = \Theta(T_\osMM(n,\ell,t)).
\end{align*}
\end{prop}
\begin{proof}
By setting $C=0$, we can reduce $\AllZeroes_t$ and $\osMM_t$ to $\MPV_t$ and $\MPC_t$, respectively, achieving the lower bounds of the claim.

For the other direction, let $a_1, \dots, a_\ell \in \ints^n$ be the row vectors of $A$, $b_1, \dots, b_\ell \in \ints^n$ be the column vectors of $B$ and $c_1, \dots, c_\ell \in \ints^\ell$ be the column vectors of $C$. Let $e_i$ denote the vector whose $i$-th coordinate is 1 and whose other coordinates are 0. We define $\ell \times (n+\ell), (n+\ell) \times \ell$ matrices $A', B'$ by specifying the row vectors of $A'$ as
\begin{align*}
  a'_i & = ( a_i, -e_i),
\end{align*}
and the column vectors of $B'$ as
\begin{align*}
  b'_j & = ( b_j, c_j).
\end{align*}
Note that $(A'B')_{i,j} = \langle a'_i, b_j' \rangle = \langle a_i, b_j \rangle - c_j[i]$, thus $(A'B')_{i,j} = 0$ if and only if $(AB)_{i,j} = C_{i,j}$. Consequently, $A'B'$ has at most $t$ nonzeroes, and checking equality of $A'B'$ to the all-zero matrix is equivalent to checking $AB = C$. The total time to solve $\MPV_t$ is thus bounded by $\Oh((n+\ell)\ell) + T_\AllZeroes(n+\ell, \ell, t) = \Oh(T_\AllZeroes(n, \ell,t))$, as desired. 

Furthermore, by computing $C' = A'B'$ (which contains at most $t$ nonzero entries), we can also correct the matrix product $C$ by updating $C_{i,j}$ to $C_{i,j} + C'_{i,j}$. This takes time $\Oh((n+\ell)\ell) + T_\osMM(n+\ell,\ell,t) = \Oh(T_\osMM(n,\ell, t))$, as desired.
\end{proof}

Because of the above equivalence, we can focus on solving $\AllZeroes_t$ and $\osMM_t$ in the remainder of the paper. The key for our approach is the following multilinear polynomial
\[ \fMM^{A,B}(x_1,\dots, x_\ell; y_1, \dots, y_\ell) := \sum_{i,j\in [\ell]} x_i \cdot y_j \cdot \langle a_i, b_j \rangle, \]
where again the $a_1,\dots, a_\ell$ denote the row vectors of $A$ and the $b_1,\dots,b_\ell$ denote the column vectors of $B$.
Note that the nonzero monomials of $\fMM^{A,B}$ correspond directly to the nonzero entries of $AB$. We introduce a univariate variant
\[ g(X) = g^{A,B}(X) := \fMM^{A,B}(1, X, \dots, X^{\ell-1}; 1, X^\ell, \dots, X^{\ell(\ell-1)}), \]
which has the helpful property that monomials $x_i y_j$ of $\fMM$ are mapped to the monomial $X^{(i-1)+\ell(j-1)}$ in a one-to-one manner, preserving coefficients. To obtain a more efficient representation of $g$ than to explicitly compute all coefficients $\langle a_i, b_j\rangle$, we can exploit linearity of the inner product: we have $g(X) = \sum_{k=1}^n q_k(X)r_k(X^\ell)$, where $q_k(Z) = \sum_{i=1}^\ell a_i[k] Z^{i-1}$ and $r_k(Z) = \sum_{j=1}^\ell b_j[k]Z^{j-1}$. This representation is more amenable for efficient evaluation, and immediately yields a reduction to univariate polynomial identity testing (\UPIT) (see Theorem~\ref{thm:MPVtoUPIT} in Section~\ref{sec:relations}). 

To solve the detection problem, we use an idea from sparse polynomial interpolation~\cite{Ben-OrT88,Zippel90}: If $AB$ has at most $t$ nonzeroes, then for any root of unity $\omega$ of sufficiently high order, $g(\omega^0) = g(\omega^1) = g(\omega^2) = \cdots = g(\omega^{t-1}) = 0$ is equivalent to $AB = 0$. By showing how to do fast batch evaluation of $g$ using the above representation, we obtain an $\tOh((\ell+ t)n)$-time algorithm for $\AllZeroes_t$ in Section~\ref{sec:errordetection}, proving Theorem~\ref{thm:main-detect}.

Towards solving the correction problem, the naive approach is to use the $\tOh((\ell + t)n)$-time $\AllZeroes_t$ algorithm in combination with a self-reduction to obtain a fast algorithm for finding a nonzero position $(i,j)$ of $AB$: If the \AllZeroes algorithm determines that $AB$ contains at least one nonzero entry, we split the product matrix $AB$ into four submatrices, detect any one of them containing a nonzero entry, and recurse on it. After finding such an entry, one can compute the correct nonzero value $(AB)_{i,j} = \langle a_i, b_j \rangle$ in time $\Oh(n)$. One can then ``remove'' this nonzero from further search (analogously to Proposition~\ref{prop:MPVtoAllZeroes}) and iterate this process. Unfortunately, this only yields an algorithm of running time $\tOh(t n^2)$, even if $\AllZeroes$ would take near-optimal time $\tOh(n^2)$. A faster alternative is to use the self-reduction such that we find \emph{all} nonzero entries whenever we recurse on a submatrix containing at least one nonzero value. However, this process only leads to a running time of $\tOh(\sqrt{t}n^2 + nt^2)$. Here, the bottleneck $\tOh(nt^2)$ term stems from the fact that performing an \AllZeroes test for $t$ submatrices (e.g., when $t$ nonzeroes are spread evenly in the matrix) takes time $t \cdot \tOh(nt)$. 

We still obtain a faster algorithm by a rather involved approach: The intuitive idea is to test submatrices for appropriately smaller number of nonzeroes $z \ll t$. At first sight, such an approach might seem impossible, since we can only be certain that a submatrix contains no nonzeroes if we test it for the full number $t$ of potential nonzeroes. However, by showing how to reuse and quickly update previously computed information after finding a nonzero, we make this approach work by obtaining ``global'' information at a small additional cost of $\tOh(t^2)$. Doing these dynamic updates quickly crucially relies on the efficient representation of the polynomial $g$. The details are given in Section~\ref{sec:errorcorrection}.

\section{Structural Results: Avenues Via Other Problems}
\label{sec:relations}

In this section, we show the simple reductions translating verifiers for \THREESUM or \UPIT to matrix multiplication.

\subsection{3SUM}

We consider the following formulation of the \THREESUM problem: given sets $S_1,S_2,S_3$ of polynomially bounded integers, determine whether there exists a triplet $s_1\in S_1, s_2 \in S_2, s_3 \in S_3$ with $s_1 + s_2 = s_3$.
It is known that a combinatorial $\Oh(n^{3/2-\varepsilon})$-time algorithm for 3SUM (for any $\varepsilon > 0$) yields a combinatorial $\Oh(n^{3-\varepsilon'})$-time Boolean matrix multiplication (\BMM) algorithm (for some $\varepsilon'>0$). This follows by combining a reduction from Triangle Detection to 3SUM of \cite{JafargholiV16} and using the combinatorial subcubic equivalence of Triangle Detection and \BMM\cite{VassilevskaWilliamsW10}\footnote{K. G. Larsen obtained an independent proof of this fact, see \url{https://simons.berkeley.edu/talks/kasper-larsen-2015-12-01}.}. While this only yields a \emph{nontight} \BMM-based lower bound for 3SUM for \emph{deterministic or randomized} combinatorial algorithms, we can establish a \emph{tight} relationship for the current state of knowledge of combinatorial verifiers. In fact, allowing nondeterminism, we obtain a very simple direct proof of a stronger relationship of the running times than known for deterministic reductions.

\begin{theorem}\label{thm:MPVto3SUM}
If 3SUM admits a (``combinatorial'') $\Oh(n^{3/2-\varepsilon})$-time verifier, then \BMM admits a (``combinatorial'') $\Oh(n^{3-2\varepsilon})$-time verifier.\footnote{Strictly speaking, the notion of a ``combinatorial'' algorithm is not well-defined, hence we use quotes here. However, our reductions are so simple that they should qualify under any reasonable exact definition.}
\end{theorem}
Thus, significant combinatorial improvements over Carmosino et al.'s \THREESUM verifier yield strongly subcubic combinatorial \BMM verifiers. In particular, a $\tOh(n)$-time verifier for 3SUM would yield an affirmative answer to our main question in the Boolean setting. Note that an analogous improvement of the $\Oh(n^{3/2}\sqrt{\log n})$~\cite{GronlundP14} size bound in the decision tree model to a size of $\Oh(n\log^2 n)$ has recently been obtained~\cite{KaneLM18}.

To establish this strong relationship, our reduction exploits the nondeterministic setting -- without nondeterminism, no reduction is known that would give a $\Oh(n^{\frac{8}{3}- \varepsilon})$-time \BMM algorithm even if \THREESUM could be solved in an optimal $\Oh(n)$ time bound.

\begin{proof}[Proof of Theorem~\ref{thm:MPVto3SUM}]
Given the $n \times n$ Boolean matrices $A,B,C$, we first check whether all entries $(i,j)$ with $C_{i,j} = 1$ are correct. For this, for each such $i,j$, we guess a witness $k$ and check that $A_{i,k} = B_{k,j} = 1$, which verifies that $C_{i,j}=(AB)_{i,j}=1$.

To check the remaining zero entries $Z= \{ (i,j)\in [n]^2  \mid C_{i,j} = 0\}$, we construct a 3SUM instance $S_1,S_2,S_3$ as follows. Let $W = 2(n+1)$. For each $(i,j) \in Z$, we include $iW^2 + jW$ in our set $S_3$. For every $(i,k)$ with $A_{i,k}=1$, we include $iW^2 + k$ in our set $S_1$, and, for every $(k,j)$ with $B_{k,j} = 1$, we include $jW - k$ in our set $S_2$. Clearly, any witness $A_{i,k}=B_{k,j}=1$ for $(AB)_{i,j}=1$, $(i,j)\in Z$ yields a triplet $a=i W^2 + k \in S_1, b= jW -k \in S_2, c=iW^2 + jW\in S_3$ with $a+b=c$. Conversely, any 3SUM triplet $a\in S_1, b\in S_2, c\in S_3$ yields a witness for $(AB)_{i,j}=1$, where $(i,j)\in Z$ is the zero entry represented by $c$, since $(iW^2 + k) + (jW - k') = i'W^2+j'W$ for $i,i',j,j',k,k'\in [n]$ if only if $i=i', j=j'$ and $k=k'$ by choice of $W$. Thus, the 3SUM instance is a NO instance if and only if no $(i,j)\in Z$ has a witness for $(AB)_{i,j} = 1$, i.e., all $(i,j)\in Z$ satisfy $C_{i,j} = (AB)_{i,j} = 0$.

Note that reduction runs in nondeterministic time $\Oh(n^2)$, using an oracle call of a 3SUM instance of size $\Oh(n^2)$, which yields the claim. 
\end{proof}

\subsection{UPIT}

Univariate Polynomial Identity Testing (\UPIT) is the following problem: Given arithmetic circuits $Q, Q'$ on a single variable, with degree $n$ and $O(n)$ wires, over a field of order $\poly(n)$, determine whether $Q \equiv Q'$, i.e., the outputs of $Q$ and $Q'$ agree on all inputs. Using evaluation on $n+1$ distinct points, we can deterministically solve \UPIT in time $\tOh(n^2)$, while evaluating on $\tOh(1)$ random points yields a randomized solution in time $\tOh(n)$. Williams~\cite{Williams16} proved that a $\Oh(n^{2-\varepsilon})$-time deterministic \UPIT algorithm refutes the Nondeterministic Strong Exponential Time Hypothesis posed by Carmosino et al.~\cite{CarmosinoGIMPS16}.  
We establish that a sufficiently strong (nondeterministic) derandomization of \UPIT also yields progress on \MPV. 

\begin{theorem}\label{thm:MPVtoUPIT}
If \UPIT admits a (``combinatorial'') $\Oh(n^{3/2-\varepsilon})$-time verifier for some $\varepsilon >0$, then there is a (``combinatorial'') $\Oh(n^{3-2\varepsilon})$-time verifier for matrix multiplication over polynomially bounded integers and over finite fields of polynomial order. \end{theorem}
\begin{proof}
We only give the proof for matrix multiplication over a finite field $\field$ of polynomial order. Using Chinese Remaindering, we can easily extend the reduction to the integer case (see Proposition~\ref{prop:findprimes} below). 

Consider $g(X) = \sum_{i,j\in [n]} \langle a_i, b_j \rangle X^{(i-1) + n(j-1)}$ over $\field$ as defined in Section~\ref{sec:mainIdeas} (with $\ell = n$). As described there, we can write $g(X) = \sum_{k=1}^n q_k(X)r_k(X^n)$ with $q_k(Z) = \sum_{i=1}^n a_i[k]Z^{i-1}$ and $r_k(Z) = \sum_{j=1}^n b_j[k]Z^{j-1}$. Let $k\in [n]$ and note that $q_k,r_k$ and $X^n$ have arithmetic circuits with $\Oh(n)$ wires using Horner's scheme. Chaining the circuits of $X^n$ and $r_k$, and multiplying with the output of the circuit for $q_k$, we obtain a degree-$\Oh(n^2)$ circuit $Q_k$ with $\Oh(n)$ wires. It remains to sum up the outputs of the circuits $Q_1, \dots, Q_n$. We thus obtain a circuit $Q$ with $\Oh(n^2)$ wires and degree $\Oh(n^2)$. Since by construction $AB = 0$ if and only $Q \equiv 0$, we obtain an \UPIT instance $Q, Q'$, with $Q'$ being a constant-sized circuit with output $0$, that is equivalent to our \MPV instance. Thus, any $\Oh(n^{3/2-\varepsilon})$-time algorithm for \UPIT would yield a $\Oh(n^{2(3/2-\varepsilon)})$-time \MPV algorithm, as desired.
\end{proof}

It is known that refuting NSETH implies strong circuit lower bounds~\cite{CarmosinoGIMPS16}, so pursuing this route might seem much more difficult than attacking \MPV directly. However, to make progress on \MPV, we only need to nondeterministically derandomize \UPIT for very specialized circuits. In this direction, our algorithmic results exploit that we can derandomize \UPIT for these specialized circuits, as long as they represent \emph{sparse} polynomials.

\section{Deterministically Detecting Presence of \texorpdfstring{$0 < z\le t$}{0< z <= t} Errors}
\label{sec:errordetection}

In this section we prove the first of our main algorithmic results, i.e., Theorem~\ref{thm:main-detect}.
\begin{theorem}\label{thm:errordetection}
For any $1 \le t \le n^2$, $\MPV_t$ can be solved deterministically in time $O( (n^2 +tn)\log^{2+o(1)} (n))$.
\end{theorem}

We prove the claim by showing how to solve the following problem in time $\tOh((\ell+t)n)$. 
\begin{lemma}\label{lem:mainalg}
Let $\field_p$ be a prime field with a given element $\omega\in \field_p$ of order at least~$\ell^2$.
Let $A,B$ be $\ell \times n, n\times \ell$-matrices over $\field_p$. There is an algorithm running in time $\Oh((\ell+t)n\log^{2+o(1)} n)$ with the following guarantees:
\begin{enumerate}
\item If $AB = 0$, the algorithm outputs ``$AB=0$''.
\item If $AB$ has $0 < z \le t$ nonzeroes, the algorithm outputs ``$AB \ne 0$''.
\end{enumerate}
\end{lemma}

Given such an algorithm working over finite fields, we can check matrix products of integer matrices using the following proposition. 
\begin{prop}\label{prop:findprimes}
Let $A,B$ be $n\times n$ matrices over the integers of absolute values bounded by $n^c$ for some $c\in \nats$. Then we can find, in time $\Oh(n^2 \log n)$, distinct primes $p_1, p_2, \dots, p_d$ and corresponding elements $\omega_1 \in \field_{p_1}, \omega_2 \in \field_{p_2}, \dots, \omega_d \in \field_{p_d}$, such that
\begin{enumerate}[label=\roman*)]
\item\label{enum:primedecomp} $AB = 0$  if and only if $AB = 0$ over $\field_{p_i}$ for all $1\le i \le d$,  
\item\label{enum:numberofprimes} $d= \Oh(1)$, and
\item for each $1\le i \le d$, we have $p_i = \Oh(n^2)$ and $\omega_i$ has order at least $n^2$ in $\field_{p_i}$.  
\end{enumerate}
\end{prop}
Note that the obvious approach of choosing a single prime field $\field_p$ with $p \ge n^{2c+1}$ is not feasible for our purposes: the best known deterministic algorithm to find such a prime takes time $n^{c/2+o(1)}$ (see~\cite{TaoCH12} for a discussion), quickly exceeding our desired time bound of $\Oh(n^2)$.

\begin{proof}[Proof of Proposition~\ref{prop:findprimes}]
Let $d = c+1$ and note that any entry $(AB)_{i,j} = \sum_{k=1}^n A_{i,k}B_{k,j}$ is in $[-n^{2c+1}, n^{2c+1}]$. Thus for any number $m > n^{2c+1}$, we have $(AB)_{ij} \equiv 0 \pmod m$ if and only if $(AB)_{i,j} = 0$. By Chinese Remaindering, we obtain that any distinct primes $p_1, \dots, p_d$ with $p_i \ge n^{2}$ satisfy~\ref{enum:primedecomp} and~\ref{enum:numberofprimes}, as $AB = 0$ if and only if $AB = 0$ over $\field_{p_i}$ for all $1\le i \le d$, using the fact that $\prod_{i=1}^d p_i \ge n^{2d} > n^{2c+1}$.

By Bertrand's postulate, there are at least $d$ primes in the range $\{n^2+1,\dots, 2^d (n^2+1)\}$, thus using the sieve of Eratosthenes, we can find $p_1, \dots, p_d$ with $p_i \ge n^2 + 1$ and $p_i \le 2^d(n^2+1)$ in time $\Oh(n^2 \log \log n)$ (see \cite[Theorem 18.10]{vzGathenG13}). It remains to find elements $\omega_1\in \field_{p_1},\dots,\omega_d \in \field_{p_d}$ of sufficiently high order. For each $1\le j \le d$, this can be achieved in time $\Oh(n^2 \log n)$ by exhaustive testing: We keep a list $L\subseteq \field_{p_j}^\times = \field_{p_j} \setminus \{0\}$ of ``unencountered'' elements, which we initially set to $\field_{p_j}^\times$. Until there are no elements in $L$ remaining, we pick any $\alpha \in L$ and delete all elements in the subgroup of $\field_{p_j}^\times$ generated by $\alpha$ from $L$. We set $\omega_j$ to the last $\alpha$ that we picked (which has to generate the complete multiplicative group $\field_{p_j}^\times$) and thus is a primitive $(p_{j}-1)$-th root of unity. Since $p_{j}-1 \ge n^2$, the order of $\omega_j$ is at least~$n^2$, as desired. 

Storing $L$ as a Boolean lookup table over $\field_{p_j}^\times$, we can perform each iteration in time $\Oh(p_j)$. Furthermore, observe that the number of iterations is bounded by the number of subgroups of~$\field_{p_j}^\times$, and thus by the number of divisors of $p_{j}-1$. Hence, we have at most $\Oh(\log p_j)$ iterations, each taking time at most $\Oh(p_j)$, yielding a running time of $\Oh(p_j \log p_j) = \Oh(n^2 \log n)$, as desired.
\end{proof}

Combining Proposition~\ref{prop:MPVtoAllZeroes} with the algorithm of Lemma~\ref{lem:mainalg} and Proposition~\ref{prop:findprimes}, we obtain the theorem.
\begin{proof}[Proof of Theorem~\ref{thm:errordetection}]
Given any instance $A,B,C$ of $\MPV_t$, we convert it to an instance $A',B'$ of \AllZeroes as in Proposition~\ref{prop:MPVtoAllZeroes}. We construct primes $p_1,\dots, p_d$ as in Proposition~\ref{prop:findprimes} in time $\Oh(n^2\log n)$. For each $j \in [d]$, we convert $A'$, $B'$ to matrices over $\field_{p_j}$ in time $\Oh(n^2)$ and test whether $A'B' = 0$ over $\field_{p_j}$ for all $j\in [d]$ using Lemma~\ref{lem:mainalg} in time $\Oh((n^2+tn)\log^{2+o(1)} n)$. We output ``$AB=C$'' if and only if all tests succeeded. Correctness follows from Proposition~\ref{prop:findprimes} and Lemma~\ref{lem:mainalg}, and the total running time is $\Oh((n^2 + tn)\log^{2+o(1)} n)$, as desired. 
\end{proof}

In the remainder, we prove Lemma~\ref{lem:mainalg}.
As outlined in Section~\ref{sec:mainIdeas}, define the polynomial $g(X) = \sum_{i,j\in [\ell]} \langle a_i, b_j \rangle X^{(i-1) + \ell(j-1)}$ over $\field_p$. We aim to determine whether $g\equiv 0$. To do so, we use the following idea from Ben-Or and Tiwari's approach to black-box sparse polynomial interpolation (see~\cite{Ben-OrT88, Zippel90}). 
Suppose that $\omega \in \field_p$ has order at least $\ell^2$. Then the following proposition holds. 
\begin{prop}\label{prop:central}
Assume $AB$ has $0 \le z\le t$ nonzeroes.
Then $g(\omega^0) = g(\omega) = g(\omega^2) = \cdots = g(\omega^{t-1}) = 0$ if and only if $g\equiv 0$, i.e., $z=0$.
\end{prop}
\begin{proof}
By assumption on $A, B$, we have $g(X) = \sum_{m \in M} c_m X^{m}$, where $M = \{ (i-1) + \ell (j-1) \mid \langle a_i, b_j \rangle \ne 0\}$ with $|M| = z \le t$ and $c_{(i-1) + \ell (j-1)} = \langle a_i, b_j \rangle$. Writing $M= \{m_1, \dots, m_z\}$ and defining $v_m = \omega^m$, we see that $g(\omega^0) = \cdots = g(\omega^{t-1}) = 0$ is equivalent to
\begin{alignat*}{4}
& c_{m_1} & +  \cdots & + c_{m_z} &&= 0,  \\
& c_{m_1}  v_{m_1} & +  \cdots & + c_{m_z} v_{m_z} && = 0,\\
& c_{m_1}  v^2_{m_1} & +  \cdots & + c_{m_z} v^2_{m_z} && = 0,\\
& &  & & &\dots \\
& c_{m_1}  v^{t-1}_{m_1} &+  \cdots & + c_{m_z} v^{t-1}_{m_z} && = 0.
\end{alignat*}
Since $\omega$ has order at least $\ell^2$, we have that $v_m = \omega^{m} \ne \omega^{m'} = v_{m'}$ for all $m, m' \in M$ with $m\ne m'$. Thus the above system is a Vandermonde system with unique solution $(c_{m_1}, \dots, c_{m_z}) = (0, \dots, 0)$, since $z\le t$. This yields the claim.
\end{proof}

It remains to compute $g(\omega^0), \dots, g(\omega^{t-1})$ in time $\tOh((\ell+t) n)$.

\begin{prop}\label{prop:fastevalg}
For any $\sigma_1, \dots, \sigma_t\in \field_p$, we can compute $g(\sigma_1), \dots, g(\sigma_t)$ in time $\Oh((\ell+t)n \log^{2+o(1)} \ell)$.
\end{prop}
\begin{proof}
Recall that $g(X) = \sum_{k=1}^n q_k(X)\cdot r_k(X^\ell)$, where $q_k(Z) = \sum_{i = 1}^{\ell} a_i[k]Z^{i-1}$ and $r_k(Z) = \sum_{j=1}^{\ell} b_j[k]Z^{j-1}$. Let $1\le k \le n$. Using fast multipoint evaluation (Lemma~\ref{lem:polymultipoint}), we can compute $q_k(\sigma_1), \dots, q_k(\sigma_t)$ using $\Oh((\ell+t)\log^{2+o(1)} \ell)$ additions and multiplications in $\field_p$. Furthermore, since we can compute $\sigma_1^\ell, \dots, \sigma_t^\ell$ using $\Oh(t\log \ell)$ additions and multiplications in~$\field_p$, we can analogously compute  $r_k(\sigma_1^\ell), \dots, r_k(\sigma_t^\ell)$ in time $\Oh((\ell+t) \log^{2+o(1)} \ell)$. Doing this for all $1\le k \le n$ yields all values $q_k(\sigma_u), r_k(\sigma_u^\ell)$ with $k\in [n],u\in [t]$ in time $\Oh((\ell+t)n\log^{2+o(1)} \ell)$. We finally aggregate these values to obtain the desired outputs $g(\sigma_u) = \sum_{k=1}^n q_k(\sigma_u)\cdot r_k(\sigma_u^\ell)$ with $u \in [t]$. The aggregation only uses $\Oh(tn)$ multiplications and additions in $\field_p$, thus the claim follows.
\end{proof}

Together with Proposition~\ref{prop:central}, this yields Lemma~\ref{lem:mainalg} and thus the remaining step of the proof of Theorem~\ref{thm:errordetection}.

\section{Deterministic Output-sensitive Matrix Multiplication}
\label{sec:errorcorrection}

In this section, we prove the second of our main algorithmic results, specifically, Theorem~\ref{thm:main-correct}.
\begin{theorem}\label{thm:errorcorrection}
We can solve $\osMM_t$ in time $\Oh(\sqrt{t}n^2 \log^{2+o(1)} n + t^2 \log^{3+o(1)} n)$.
\end{theorem}

Recall that $A,B$ are $n\times n$ matrices, where $A$ has rows $a_1, \dots, a_n$ and $B$ has columns $b_1,\dots,b_n$. Without loss of generality, we assume that $n$ is a power of two. Furthermore, for ease of presentation, we only consider computing $AB$ over a prime field $\field_p$ with $p = \Theta(n^2)$ and a given element $\omega \in \field_p$ of order at least $n^2$. Using Proposition~\ref{prop:findprimes}, it is straightforward to adapt our approach to work for polynomially bounded integer matrices as well, analogously to the proof of Theorem~\ref{thm:errordetection}.

\paragraph*{Iterative matrix structure}
The algorithm will sequentially find nonzero entries, compute the correct values to record them in the result matrix $C$ and repeat until all nonzeroes are found. To ``remove'' already found nonzeroes from $AB$, we define (as in Proposition~\ref{prop:MPVtoAllZeroes}) the $n\times 2n, 2n\times n$ matrices $A',B'$ (depending on $A,B$ and the current state of $C$) by
\[A' = \left(\begin{array}{cc} A & C\end{array}\right), B' = \left(\begin{array}{c} B \\ -I\end{array}\right).\]
Let $a'_i$ be the $i$-th row of $A'$ and $b'_j$ the $j$-th column of $B'$, then $\langle a'_i, b'_j \rangle \ne 0$ if and only if $(AB)_{i,j} \ne C_{i,j}$. In particular, consider the following process after initializing $C \leftarrow 0$: Until $A'B' = 0$, we find any nonzero entry $(A'B')_{i,j}\ne 0$, and set $C_{i,j} = (AB)_{i,j}$. If $AB$ has $z$ nonzero entries, this process terminates after $z$ iterations with $C = AB$.

\paragraph*{Canonical submatrices}
We operate on \emph{submatrices} of $A'B'$ specified by an interval $I \subseteq [n]$ of rows of $A'$ and an interval $J \subseteq [n]$ of columns of $B'$. Frequently, we write an even-sized interval $I$ as the disjoint union $I_1\cup I_2$, where $I_1$ denotes the half of smaller elements and $I_2$ denotes the half of larger elements. For any submatrix $I,J$ with $|I|=|J|=2^\kappa$ for some integer $\kappa \ge 1$, we call $I_a,J_b$ with $a,b\in \{1,2\}$ a \emph{child submatrix} of $I,J$. Correspondingly, $I,J$ is called the \emph{parent submatrix} of $I_a,J_b$ for $a,b\in \{1,2\}$. We say that $I,J$ is a \emph{canonical submatrix} if $I=J=[n]$ or $I,J$ is a child submatrix of a canonical submatrix.

\paragraph*{A first failed approach}
Using the \AllZeroes test of Lemma~\ref{lem:mainalg}, we can check whether $A'B'$ contains a nonzero in time $\tOh(n^2 + tn)$. If this is the case, we can detect some of the four child submatrices of $[n],[n]$ containing at least one nonzero in time $\tOh(n^2 + tn)$ and recurse on it. In this way we can find a single nonzero in time $\tOh(n^2+tn)$, yielding a $\tOh(tn^2 + t^2 n)$-time solution to compute all nonzeroes, which is much slower than our desired bound. To improve upon this running time, we introduce the notion of \emph{test values} for submatrices and show how to reuse test values computed in a previous iteration.

\paragraph*{Test values}
For any canonical submatrix $I= \{i+1,\dots, i+\ell\}, J = \{j+1,\dots, j+\ell\}$, we define
\begin{align*}
q^{I}_k(Z) & = \sum_{s=1}^\ell a'_{i+s}[k] Z^{s-1}, & & \text{where } k\in [2n], \\
r^{J}_k(Z) & = \sum_{s=1}^\ell b'_{j+s}[k] Z^{s-1}, & & \text{where } k\in [2n], \\
g^{I,J}(X) & = \sum_{k=1}^{2n} q^{I}_k(X) \cdot r^{J}_k(X^\ell).
\end{align*}
Recall that Proposition~\ref{prop:central} yields that for any $\omega \in \field_p$ of order at least $n^2$, we have 
\begin{enumerate}[label=(\Roman{*})]
\item \label{enum:soundness} \emph{soundness}: if  $g^{I,J}(\omega^{\nu}) \ne 0$ for some $\nu$, then $(A'B')_{I,J}$ contains at least one nonzero entry, and
\item \label{enum:completeness} \emph{completeness}: if the submatrix $(A'B')_{I,J}$ contains $z>0$ nonzeroes, then there is some $0 \le \nu < z$ such that  $g^{I,J}(\omega^{\nu}) \ne 0$. 
\end{enumerate}

Let $I,J$ be any canonical submatrix of $A'B'$ and $\tau\ge 1$. we call the $\tau$ values $\gamma^{I,J}_0 = g^{I,J}(\omega^{0}), \dots, \gamma^{I,J}_{\tau-1} = g^{I,J}(\omega^{\tau-1})$ the \emph{test values for $I$,$J$ at granularity $\tau$}. We assign to each canonical submatrix $I,J$ a granularity $\gamma^{I,J}$ that is initialized to 0. In a certain sense, $\tau^{I,J}$ is a guess on the number of nonzeroes in $(A'B')_{I,J}$. During the process, we will take care to always maintain the test values of all four child submatrices of $I,J$ at granularity $\tau^{I,J}$, i.e., values $\gamma^{I_a,J_b}_\nu = g^{I_a,J_b}(\omega^\nu)$ for all $0 \le \nu < \tau^{I,J}$, $a,b\in \{1,2\}$, even after updating the matrix $C$ (and thus $A'$). Note that by Proposition~\ref{prop:fastevalg}, for any canonical submatrix $I,J$, we can compute test values for $I, J$ at granularity $\tau$ in time $\tOh( (|I|+\tau) n)$. 

\paragraph*{A second failed approach}
A natural idea is to find \emph{all} nonzeroes in $(A'B')_{I,J}$ once we determine that a canonical submatrix $I,J$ contains at least one nonzero. This can be done by performing an \AllZeroes test on all four child submatrices of $I,J$ in time $\tOh( (|I|+t)n)$ and recursing on all those children containing at least one nonzero. It is straightforward to show that this amounts to an algorithm running in time $\tOh(\sqrt{t} n^2 + t^2 n)$, still slower than our desired running time.

One might try to use the following observation: Let $z$ be the number of nonzeroes of $(A'B')_{I,J}$. Then in fact already the test values at granularity $z$ would successfully detect all those child submatrices containing at least one nonzero. This might seem to yield a faster test  at this level of this recursion with a running time $\tOh( (|I|+z)n)$ instead of $\tOh( (|I|+t)n)$. If this was indeed possible, then this would yield a $\tOh(\sqrt{t}n^2 + tn)$-time algorithm (thus, a faster algorithm than what we provide). However, the exact value of $z$ is unknown -- if some child submatrix has only zeroes as test values at granularity $\tau \ll t$, then it still might have nonzeroes for larger granularities, i.e., we do not know when \emph{not} to recurse on a child submatrix without testing at full granularity $t$.

Surprisingly, we can still remedy this situation by incurring an additional running time cost of $\tOh(t^2)$. The high-level idea is as follows: Once we determine a submatrix to contain a nonzero, we do an exponential search for the lowest granularity at which we can find a nonzero entry. The crucial point is to obtain a stopping criterion: we show how to dynamically update all previously computed test values when we ``remove'' another nonzero from the search in time $\tOh(t)$. Intuitively, this allows us to determine when to stop the recursion on some submatrix. This update heavily depends on the specific structure of the polynomials~$g^{I,J}$. In the remainder, we give the full description and analysis of this approach.

\paragraph*{Submatrix queue}
As an invariant, we maintain a list $L$ of submatrices $I,J$ with the property that \emph{$(A'B')_{I,J}$ contains at least one nonzero entry}.  Until all nonzeroes are found, it will contain $[n],[n]$, and each member of this list is a submatrix of the preceding member of this list. We iteratively take the last (i.e., smallest) submatrix in $L$ and find some nonzero position $i\in I, j\in J$ of $(A'B')_{I,J}$. We update $C_{i,j}$ and update test values such that for each canonical submatrix $(I,J)$, all test values at the corresponding granularity are kept up-to-date. 

Algorithm~\ref{alg:main} gives the formal outline of the algorithm.

\begin{algorithm} 
\begin{algorithmic}[1]
\Function{\osMM}{$A$,$B$,$t$}
\State initialize $C\gets 0$
\State compute test values $\gamma^{[n],[n]}_0, \dots, \gamma^{[n],[n]}_{t-1}$
\If{$\gamma^{[n],[n]}_\nu \ne 0$ for some $0 \le \nu< t$ }
\State add $[n], [n]$ to $L$
\EndIf
\While{ $L$ is not empty } \label{line:while}
\State let $(I,J)$ be the smallest submatrix in $L$
\State $(i,j) \gets \FindNonzero(I,J)$
\State $C_{i,j} \gets \langle a_i, b_j \rangle$
\State $\UpdateVals([n],[n],i,j)$
\EndWhile
\State \Return $C$
\EndFunction
\end{algorithmic}
\caption{Computing the matrix product $AB$, if $AB$ contains at most $t$ nonzeroes.}
\label{alg:main}
\end{algorithm}

\paragraph*{Finding a Nonzero}
By the above outline, we only call $\FindNonzero$ on submatrices for which we know that there is at least one nonzero. We split each matrix into four equi-dimensional submatrices and do an exponential search for the smallest granularity such that the test values of the submatrices allow us to determine a submatrix containing at least one nonzero. Note that here, we only compute test values if they have not previously been computed. Furthermore, when we compute test values for the first time, we compute test values for the granularity $\tau = |I|=|J|$ (since computing test values for submatrix $I,J$ of granularity $\tau$ takes time $\tOh((|I|+\tau)n)$, we obtain the first $|I|$ test values essentially for free).

The exponential search guarantees that we never set $\tau^{I,J}$ to a value higher than $2z$, where $z$ denotes the number of nonzeroes in $(AB)_{I',J'}$ (see Lemma~\ref{lem:timeprops}\ref{enum:granularitybound}).  

The formal outline is given in Algorithm~\ref{alg:FindNonzero}.

\begin{algorithm} 
\begin{algorithmic}[1]
\Function{\FindNonzero}{$I$,$J$}
\If{$I=\{i\}$ and $J=\{j\}$}
\Return $(i,j)$
\EndIf
\State split $I = I_1 \cup I_2$, $J=J_1 \cup J_2$ into equal-sized disjoint intervals
\If{$\tau^{I,J} = 0$} $\tau^{I,J} \gets |I|$
\EndIf
\For{all $a,b \in \{1,2\}$}
\State compute and store test values for $I_a,J_b$ at granularity $\tau^{I,J}$, \emph{if not already stored}
\EndFor
\If{$\gamma^{I_a,J_b}_\nu \ne 0$ for some $a,b\in \{1,2\}, 0 \le \nu < \tau^{I,J}$}
\State add $I_a,J_b$ to $L$
\State \Return $\FindNonzero(I_a,J_b)$
\Else
\State $\tau^{I,J} \gets 2\tau^{I,J}$
\State \Return $\FindNonzero(I,J)$
\EndIf
\EndFunction
\end{algorithmic}
\caption{Subroutine to find a nonzero entry in submatrix $(A'B')_{I,J}$.}
\label{alg:FindNonzero}
\end{algorithm}

\paragraph*{Updating Test Values}
Crucially, we rely on being able to quickly update test values once we have determined some nonzero entry and update our result. Naively recomputing the test values at full granularity~$t$ already for a single submatrix costs at least $\Omega(tn)$ time, which would yield a total update time of $\Omega(nt^2)$. To avoid these costs, we use the observation that after updating a single entry $C_{i,j}$ of $C$ (and thus $A'B'$), the only change of test values affect $\gamma^{I,J}_\nu$ with $i\in I$, $j\in J$, as the only change in the polynomials $q^{I}_k, r^{J}_k$ concerns a single coefficient change of~$q^I_{n+j}$ for the intervals $I$ with $i\in I$ (see Lemma~\ref{lem:changes}). 

To formalize the update rule, let $(i,j)$ be a position of $C$ that we set to a nonzero value. Note that this changes $A'$ by changing the coordinate $A'_{i,n+j}$ from 0 to $C_{i,j}$. Correspondingly, we define for any $I$, $\bar{q}^I_{n+j}(Z)$ as the polynomial for the old values (i.e., where the coefficient corresponding to $A'_{i,n+j}$ is 0, while $q^I_{n+j}$'s corresponding coefficient is $C_{i,j}$). We then update the test values as specified in Algorithm~\ref{alg:update}.

\begin{algorithm} 
\begin{algorithmic}[1]
\Function{\UpdateVals}{$I$,$J$,$i$,$j$}
\State set $\tau = \begin{cases} t & \text{if } I=J=[n],\\ \tau^{I',J'} & \text{o.w., where } I',J' \text{ is the parent submatrix of } I,J\end{cases}$
\State compute $\bar{q}_{n+j}^{I}(\omega^{\nu})$, $q_{n+j}^{I}(\omega^{\nu})$ and $r_{n+j}^{J}(\omega^{\nu})$ for all $0 \le \nu < \tau$, using Proposition~\ref{prop:fastevalg}
\For{$0 \le \nu < \tau$}
\State update $\gamma^{I,J}_\nu \gets \gamma^{I,J}_\nu + (q^I_{n+j}(\omega^{\nu}) - \bar{q}^I_{n+j}(\omega^\nu)) \cdot r_{n+j}^J(\omega^\nu)$ 
\EndFor
\If{$\gamma^{I,J}_\nu = 0$ for all $0 \le \nu < \tau$}
delete $I,J$ from $L$
\EndIf
\If{$|I| > 1$}
\State split $I = I_1 \cup I_2$, $J=J_1 \cup J_2$ into equal-sized disjoint intervals
\State let $a,b\in \{1,2\}$ such that $i \in I_a, j\in J_b$
\State $\UpdateVals(I_a, J_b, i, j)$
\EndIf
\EndFunction
\end{algorithmic}
\caption{Updating test values and $L$-membership of all canonical submatrices of $I,J$ after a value change in $C_{i,j}$}
\label{alg:update}
\end{algorithm}

\subsection{Correctness}

Let us argue that the output of Algorithm~\ref{alg:main} satisfies $C=AB$ whenever $AB$ contains at most $t$ nonzeroes. Note that any $C_{i,j}$ set to a nonzero value during the process is set to its correct value $\langle a_i, b_j \rangle$. It remains to argue that we indeed find all positions $(i,j)$ of nonzeroes.

Let us first consider updates of test values. Crucially, we establish that after every update of some $C_{i,j}$, the call $\UpdateVals([n],[n],i,j)$ correctly updates all previously computed test values $\gamma^{I,J}_\nu$ to maintain $\gamma^{I,J}_\nu = g^{I,J}(\omega^{\nu})$. This follows from the following observation.
\begin{lemma}\label{lem:changes}
Consider a change of $C_{i,j}$ from 0 to $\langle a_i, b_j \rangle \ne 0$. Then the only test values~$\gamma^{I,J}_\nu$ that change satisfy $i\in I, j\in J$. In particular, for the resulting matrices $A',B'$ we have $g^{I,J}(\omega^\nu) =  \bar{g}^{I,J}(\omega^\nu) + (q^I_{n+j}(\omega^{\nu}) - \bar{q}^I_{n+j}(\omega^\nu)) \cdot r_{n+j}^J(\omega^\nu)$, where $\bar{g}^{I,J}, \bar{q}^I_{n+j}$ are the polynomials $g^{I,J}, q^{I}_{n+j}$ before the change.
\end{lemma}
\begin{proof}
Note that a change of $C_{i,j}$ is a change of $A'_{i, n+j}$, whose values is only used as a coefficient for the monomial representing $i$ in the polynomials $q^I_{n+j}$ with $i\in I$. Furthermore, by definition of $B'$, we have $r^J_{n+j} \not \equiv 0$ if and only if $j \in J$ (if $j \in J$, $r^J_{n+j}$ consists of a single monomial with coefficient -1, representing $j$; otherwise, $r^J_{n+j} \equiv 0$). By the two facts above, we have $g^{I,J}(\omega^\nu) - \bar{g}^{I,J}(\omega^\nu) = (q^I_{n+j}(\omega^{\nu}) - \bar{q}^I_{n+j}(\omega^\nu)) \cdot r_{n+j}^J(\omega^\nu)$, which can be nonzero only if $i\in I$ and $j\in J$.
\end{proof}

Thus, whenever $(A'B')$ contains $0 < z \le t$ nonzeroes, the list $L$ cannot be empty: By completeness~\ref{enum:completeness}, some $\gamma^{[n],[n]}_\nu$ with $0 \le \nu < t$ must be nonzero. Since we keep all test values correctly updated, $[n],[n]$ will be removed from this list at the time \emph{all} nonzeroes have been found and removed. 

Furthermore, we maintain the invariant that all submatrices in $L$ and all submatrices for which we call $\FindNonzero$ indeed contain at least one nonzero: By soundness~\ref{enum:soundness}, we add $I,J$ to $L$ and call $\FindNonzero(I,J)$ only if $(A'B')_{I,J}$ indeed contains at least one nonzero. We remove $I,J$ from $L$ once the test values no longer guarantee $(A'B')_{I,J}$ to contain a nonzero.  

Finally, $\FindNonzero(I,J)$ terminates, yielding a nonzero entry: Since $(A'B')_{I,J}$ contains $0 < z \le t$ nonzeroes, it either consists of a single nonzero entry or has a submatrix $I_a, J_b$, $a,b\in \{1,2\}$ with at most $z$ nonzeroes. At the latest when we make a call $\FindNonzero(I,J)$ with granularity $\tau^{I,J} \ge z$, some test value $\gamma^{I_a,J_b}_\nu, 0 \le \nu < \tau^{I,J}$ must be nonzero by completeness \ref{enum:completeness}, and we recurse on a smaller subproblem.

Combining the observations above, we obtain that as long as not all nonzeroes have been found and removed, the while loop of Line~\ref{line:while} in Algorithm~\ref{alg:main} will make another iteration that finds, correctly determines  and removes a nonzero entry, yielding correctness of the algorithm.

\subsection{Running Time}

We bound the running time of Algorithm~\ref{alg:main} by $\tOh(\sqrt{t}n^2 + t^2)$. To this end, we start with a few central observations.

\begin{lemma}\label{lem:timeprops}
Algorithm~\ref{alg:main} has the following properties.
\begin{enumerate}[label=\roman{*})]
\item\label{enum:granularitybound} Let $(I,J)$ be a canonical submatrix and $z$ be the number of nonzeroes of $(AB)_{I,J}$. At the end of the process, we have $\tau^{I,J} \le \max\{|I|, 2z\}$.
\item\label{enum:totalfindcalls} The total running time of calls to $\FindNonzero$ is bounded by $\Oh( \sqrt{t}n^2\log^{2+o(1)} n + tn\log^{3+o(1)} n)$.
\end{enumerate}
\end{lemma}
\begin{proof}

For~\ref{enum:granularitybound}, assume that $\tau^{I,J} = \tau > |I|$. In this case, there must have been some call $\FindNonzero(I,J)$ which increased $\tau^{I,J}$ from $\tau/2$ to $\tau$. Consider the last such call. Let $z$ be the number of nonzeroes of $(A'B')_{I,J}$ and recall that we call $\FindNonzero(I,J)$ only if $z>0$. Thus, there must be some child submatrix $I_{a},J_{b}$ of $I,J$ with $0 < z' \le z$ nonzeroes in $(A'B')_{I_{a},J_{b}}$. Since the current call increases $\tau^{I,J}$ to $\tau$, all test values at granularity $\tau/2$ must be zero, in particular also the test values for $I_{a},J_{b}$.  At this point, $(A'B')_{I_{a},J_{b}}$ must have at least $\tau/2+1$ nonzeroes, since if $z' \le \tau/2$, the test values at granularity $\tau/2$ would have detected at least one nonzero for $I_{a},J_{b}$ by completeness~\ref{enum:completeness}. Thus $z \ge z' \ge \tau/2 + 1$, yielding the claim since the number of nonzeroes of $(AB)_{I,J}$ is never less than the number $z$ of nonzeroes in $(A'B')_{I,J}$. 

For~\ref{enum:totalfindcalls}, we first specify more precisely how we implement Algorithm~\ref{alg:FindNonzero}: Consider any child submatrix $I_a,J_b$ for which we compute test values by a call $\FindNonzero(I,J)$ of the parent submatrix. We keep a counter $\alpha^{I,J} \le \tau^{I,J}$ that stores the highest granularity $\tau$ for which we have computed test values $\gamma^{I_a,J_b}_0, \dots, \gamma^{I_a,J_b}_{\tau-1}$, where $a,b\in \{1,2\}$ (initially, $\alpha^{I,J} = 0$). In this way, we can quickly determine whether we have already all desired test values in store or need to compute additional test values. Specifically, whenever we need to compute new test values, i.e. $\alpha^{I,J} < \tau^{I,J}$, we use Proposition~\ref{prop:fastevalg} to compute the $\tau^{I,J}- \alpha^{I,J}$ missing test values $\gamma^{I_a,J_b}_{\alpha^{I,J}}, \dots, \gamma^{I_a,J_b}_{\tau^{I,J}-1}$ in time $\Oh((|I|+\tau^{I,J}-\alpha^{I,J}) n\log^{2+o(1)} n) = \Oh((\tau^{I,J}-\alpha^{I,J})n\log^{2+o(1)} n)$ (note that $\tau^{I,J} - \alpha^{I,J} \ge |I|$ holds since $\tau^{I,J}$ is initially set to $|I|$, and always doubled afterwards). In total, the total time spent for computing test values $\gamma^{I_a,J_b}_0, \dots, \gamma^{I_a,J_b}_{\tau^{I,J}-1}$, \emph{disregarding the time spent in later updates}, is bounded by $\Oh(\tau^{I,J} n \log^{2+o(1)} n)$. 
To store the test values, we maintain a list that stores $(\nu,\gamma^{I_a,J_b}_{\nu})$ (sorted by $\nu$) for all $0 \le \nu < \alpha^{I,J}$ with $\gamma^{I_a,J_b}_\nu \ne 0$. In this way, we can determine in time $\Oh(1)$ whether some nonzero test value exists, and still recover all test values.

For the analysis, we build a tree over submatrices $I,J$ for which $\FindNonzero(I,J)$ was called at least once. We assign to each node $I,J$ the total time spent in calls $\FindNonzero(I,J)$, without counting the time spent in recursive calls $\FindNonzero(I_a,J_b)$ to smaller submatrices $I_a,J_b$. In this tree,  $I,J$ is a parent of $I_a,J_b$ if any $\FindNonzero(I,J)$ call resulted in a call $\FindNonzero(I_a,J_b)$. Observe that $[n],[n]$ is the root node; we call $I,J$ a level-$i$ node, if its distance to $[n],[n]$ is $i$. Note that a level-$i$ node has $|I|=|J|=n/2^i$. We first argue that the total time spent for a level-$i$ node $I,J$ is bounded by $\Oh( (n/2^i+z(I,J))n\log^{2+o(1)} n)$, where $z(I,J)$ is the number of nonzeroes in $(AB)_{I,J}$: We account for the computation of the test values (again, disregarding updates) by $\Oh(\tau^{I,J} n \log^{2+o(1)} n)$ (as argued above). Checking for nonzero test values takes constant time per call (as argued above), and there are at most $z(I,J)$ many calls that result in determining a submatrix containing a nonzero and at most $\log \tau^{I,J}$ calls that result in doubling the granularity. Thus, the running time assigned to $I,J$ is bounded by $\Oh( \tau^{I,J} n \log^{2+o(1)} n + z(I,J) + \log \tau^{I,J}) = \Oh( (n/2^i + z(I,J)) n \log^{2+o(1)} n)$, using that $\tau^{I,J} = \Oh(|I|+z(I,J))$ by \ref{enum:granularitybound}.

To bound the total running time of all calls to \FindNonzero, note that we have at most $\min\{2^{2i}, t\}$ nodes at level $i$ (since there are only $2^{2i}$ such submatrices, and the path of each of the at most $t$ nonzeroes can contribute at most one node in each level). Define $\bar{i} = (1/2) \cdot \log_2 t$ and compute,
\begin{align*} 
\sum_{i=0}^{\log_2 n} \sum_{\text{level-}i \; (I,J)} (n/2^i + z(I,J)) & \le  \sum_{i=0}^{\log_2 n} \min\{2^{2i}, t\} n/2^i + \sum_{i=0}^{\log_2 n} \sum_{\text{level-}i \;(I,J)} z(I,J) \\
& \le  \sum_{i=0}^{\bar{i}}  2^i n + \sum_{i=\bar{i}+1}^{\log_2 n} tn/2^i + \sum_{i=0}^{\log_2 n} \sum_{\text{level-}i \; (I,J)} z(I,J) \\
& \le  n 2^{\bar{i}+1} + tn/2^{\bar{i}} + \sum_{i=0}^{\log_2 n} t = \Oh(n\sqrt{t} + t\log n). 
\end{align*}
Thus, we obtain that the total running time spent in calls to $\FindNonzero$ is bounded by
\[ \Oh\left(\sum_{i=0}^{\log_2 n} \sum_{\text{level-}i \; (I,J)} (n/2^i + z(I,J)) n \log^{2+o(1)} n\right) = \Oh(\sqrt{t} n^2  \log^{2+o(1)} n + tn\log^{3+o(1)} n). \]
\end{proof}

Updating test values can be done in time $\tOh(t+n)$ per update.
\begin{lemma}\label{lem:updatetime}
A call $\UpdateVals(I,J, i, j)$ with $|I|=|J|=\ell$ runs in time $\Oh( (\ell+t) \log^{3+o(1)} \ell)$.
\end{lemma}
\begin{proof}
Note that in the execution of $\UpdateVals(I,J,i,j)$, we have $\tau \le 2t$ by Lemma~\ref{lem:timeprops}~\ref{enum:granularitybound}. Thus, computing all values $\bar{q}^I(\omega^\nu), q^I(\omega^\nu), r^J(\omega^\nu), 0 \le \nu < \tau$ takes time $\Oh( (\ell+t) \log^{2+o(1)} \ell)$ by Proposition~\ref{prop:fastevalg}. Recall that all test values are stored as a list of nonzero values (as detailed in the proof of Lemma~\ref{lem:timeprops}). For each $\nu = 0, \dots, \tau -1$, we see whether a nonzero value $\gamma^{I,J}_\nu$ is stored in the list, otherwise $\gamma^{I,J}_\nu = 0$. We can thus compute the new value of $\gamma^{I,J}_\nu$ and, update the list with the new value (i.e., include it if it is nonzero, and leave it out if zero). This takes time $\Oh(t)$. Since we do these $\Oh((\ell+t) \log^{2+o(1)} \ell)$-time computations in each of the $\log_2 \ell$ levels of the recursion, the total running time is bounded by $\Oh( (\ell+t) \log^{3+o(1)} \ell)$.
\end{proof}

Note that in Algorithm~\ref{alg:main}, we spent a time of $\Oh( (n+t) \log^{2+o(1)} n)$ to compute the test values for $[n],[n]$ using Proposition~\ref{prop:fastevalg}. Afterwards, we have at most $t$ calls of the form $\UpdateVals([n],[n], i, j)$ with a cost of $\Oh( (n+t) \log^{3+o(1)} n )$ each (Lemma~\ref{lem:updatetime}), plus at most $t$ computations of nonzero entries in $C$ with a cost of $\Oh(n)$ each, plus the cost of all calls to \FindNonzero which amounts to $\Oh(\sqrt{t}n^2 \log^{2+o(1)} n + tn \log^{3+o(1)} n)$ by Lemma~\ref{lem:timeprops}~\ref{enum:totalfindcalls}. Thus, in total we obtain $\Oh( t (n+t) \log^{3+o(1)} n + tn + \sqrt{t}n^2 \log^{2+o(1)} n) =  \Oh( t^2 \log^{3+o(1)} n + \sqrt{t}n^2 \log^{2+o(1)} n)$, where we used that $tn = \Oh(t^2 + n^2)$.

This completes the analysis of the algorithm, and thus the proof of Theorem~\ref{thm:errorcorrection}.

\section{Open Questions}
\label{sec:open}

It remains to answer our main question. To this end, can we exploit any of the avenues presented in this work? In particular: Can we (1) find a faster \THREESUM verifier, (2) find a faster \UPIT algorithm for the circuits given in Theorem~\ref{thm:MPVtoUPIT}, or (3) instead of derandomizing Freivalds' algorithm, nondeterministically derandomize the sampling-based algorithm following from our main algorithmic result (which detects up to $\Oh(n)$ errors using Theorem~\ref{thm:main-detect}, and then samples and checks $\Theta(n)$ random entries)?

A further natural question is whether we can use the sparse polynomial interpolation technique by Ben-Or and Tiwari~\cite{Ben-OrT88} (see also \cite{Zippel90, KaltofenY88} for alternative descriptions of their approach) to give a more efficient deterministic algorithm for output-sensitive matrix multiplication. Indeed, they show how to use $\Oh(t)$ evaluations of a $t$-sparse polynomial $p$ to efficiently interpolate $p$ (for $p=g^{A,B}$, this corresponds to determining $AB$). Specifically, the $\Oh(t)$ evaluations define a certain Toeplitz system whose solution yields the coefficients of a polynomial $\zeta(Z) = \prod_{i=1}^z (Z-r_i)$ where $r_i$ is the value of the $i$-th monomial of $p$ evaluated at a certain known value. By factoring $\zeta$ into its linear factors, we can determine the monomials of $p$ (i.e., for $p=g^{A,B}$, the nonzero entries of $AB$). In our case, we can then obtain $AB$ by naive computations of the inner products at the nonzero positions in time $\Oh(nt)$. The bottleneck in this approach appears to be deterministic polynomial factorization into linear factors: In our setting, we would need to factor a degree-$(\le t)$ polynomial over a prime field~$\field_p$  of size $p = \Theta(n^2)$. We are not aware of deterministic algorithms faster than Shoup's $\Oh( t^{2+\varepsilon} \cdot \sqrt{p}\log^2 p)$-time algorithm~\cite{Shoup90}, which would yield an $\Oh(n^2 + nt^{2+\varepsilon})$-time algorithm at best. However, such an algorithm would be dominated by Kutzkov's algorithm~\cite{Kutzkov13}. Can we sidestep this bottleneck? Note that some works improve on Shoup's running time for suitable primes (assuming the Extended Riemann Hypothesis; see \cite[Chapter 14]{vzGathenG13} for references).

\section*{Acknowledgements}
The author wishes to thank Markus Bl\"aser, Russell Impagliazzo, Kurt Mehlhorn, Ramamohan Paturi, and Michael Sagraloff for early discussions on this work and Karl Bringmann for comments on a draft of this paper.

\bibliography{matmult}

\newpage

\appendix 

\section{A Note on Earlier Work}
\label{sec:errordescription}

We found work~\cite{Wiedermann14} that claims to have solved our main question in the affirmative. Unfortunately, the approach is flawed -- we detail the issue here for completeness.
The approach pursued in~\cite{Wiedermann14} is as follows: Defining $D := AB - C$, the aim is to check whether $D = 0$. Define 
\begin{align*}
x(r) & := \left(\begin{array}{c} 1 \\ r \\ r^2 \\ \vdots \\ r^{n-1}\end{array} \right), & \text{and} & & p(r):= x(r)^T \cdot D \cdot x(r).
\end{align*}
The author claims that $p \equiv 0$ if and only if $D=0$. If this would hold, one could evaluate~$p$ on a small number of points to determine whether $p\equiv 0$ and thus $D = 0$. However, the claim does not hold: note that $p(r) = \sum_{i=2}^{2n} \left(\sum_{j=1}^{i-1} D_{j,i-j}\right) r^{i-2}$. Thus already the matrix     
\[ D = \left(\begin{array}{cc} 0 & 1\\ -1 & 0,\end{array}\right)\]
satisfies $D\ne 0$, but $p(r) = 0$.

\end{document}